\documentclass[sn-mathphys-num]{sn-jnl}% Math and Physical Sciences Numbered Reference Style
\usepackage{graphicx}%
\usepackage{multirow}%
\usepackage{subfigure}
\usepackage{amsmath,amssymb,amsfonts}%
\usepackage{amsthm}%
\usepackage{mathrsfs}%
\usepackage[title]{appendix}%
\usepackage{xcolor}%
\usepackage{textcomp}%
\usepackage{manyfoot}%
\usepackage{booktabs}%
\usepackage{algorithm}%
\usepackage{algorithmicx}%
\usepackage{algpseudocode}%
\usepackage{listings}%
\usepackage{cleveref}
\theoremstyle{thmstyleone}%
\newtheorem{theorem}{Theorem}%  meant for continuous numbers
\theoremstyle{thmstyletwo}%

\theoremstyle{thmstylethree}%

\begin{document}

\title[NVEP is NPC]{The $n$-vehicle exploration problem is NP-complete}


\author*[1]{\fnm{Jinchuan} \sur{Cui}}\email{cjc@amss.ac.cn}

\author[2]{\fnm{Xiaoya} \sur{Li}}\email{xyli@amss.ac.cn}
%%\equalcont{These authors contributed equally to this work.}

%%\author[1,2]{\fnm{Third} \sur{Author}}\email{iiiauthor@gmail.com}

\affil[1,2]{\orgdiv{Academy of Mathematics and Systems Science}, \orgname{Chinese Academy of Sciences}, \orgaddress{\street{No. 55 Zhongguancun East Road}, \city{Beijing}, \postcode{100190}, \country{China}}}

%\affil[2]{\orgdiv{Academy of Mathematics and Systems Science}, \orgname{Chinese Academy of Sciences}, \orgaddress{\street{No. 55 Zhongguancun East Road}, \city{Beijing}, \postcode{100190}, \country{China}}}


\abstract{The $n$-vehicle exploration problem (NVEP) is a nonlinear unconstrained optimization problem. Given a fleet of $n$ vehicles with mid-trip refueling technique, the NVEP tries to find a sequence of $n$ vehicles to make one of the vehicles travel the farthest, and at last all the vehicles return to the start point. NVEP has a fractional form of objective function, and its computational complexity of general case remains open. Given a directed graph $G$, it can be reduced in polynomial time to an instance of NVEP. We prove that the graph $G$ has a hamiltonian path if and only if the reduced NVEP instance has a feasible sequence of length at least $n$. Therefore we show that Hamiltonian path $\leq_P$ NVEP, and consequently prove that NVEP is NP-complete.}



%\pacs[MSC Classification]{68Q25, 68Q17}

\maketitle

\section{Introduction}
\label{intro}

The $n$-vehicle exploration problem (NVEP) is motivated by a problem in a contest puzzle in China \cite{shuxue}: there are $n$ vehicles $V_1, \cdots, V_n$ with fuel capacities $a_i$ and fuel consumption rates $b_i$ for $1\leqslant i \leqslant n$. The $n$ vehicles depart from the same point to a same target at a same rate. Without fuel supply in the middle, but each vehicle can stop wherever to refuel other vehicles instantaneously, and then return to the start point. The goal is to determine a sequence $\pi = (\pi(1), \cdots, \pi(n))$ to make one of the vehicles travel the maximal distance, and at last all the vehicles return to the start point?

Consider a permutation order $\pi$ and its related sequence $V_{\pi(1)} \Rightarrow \cdots \Rightarrow V_{\pi(n)}$, where $V_{\pi(i)}$ refuels to $V_{\pi(j)}$ for any $i < j$. Let $D_{\pi} = \sum\limits_{i=1}^n{d_{\pi(i)}}$ denotes the maximal distance that the sequence can approach. $D_{\pi}$ consists of $n$ segmented parts as $d_{\pi(i)}$ which denotes the distance that $V_{\pi(i)}$ contributes to the total distance.

\begin{equation}\label{eqt1}
(NVEP_{LP})
\begin{array}{lll}
\max\ D_{\pi} = \sum\limits_{i=1}^n{d_{\pi(i)}}\\
s.t.\left\{\begin{array}{lll} & 2d_{\pi(n)}b_{\pi(n)} \leq a_{\pi(n)}\\
               & 2d_{\pi(n-1)}(b_{\pi(n-1)} + b_{\pi(n)}) + 2d_{\pi(n)}b_{\pi(n)}
\leq a_{\pi(n-1)} + a_{\pi(n)}\\
                  & \cdots\cdots \\
                   & 2d_{\pi(1)}(b_{\pi(1)} + b_{\pi(2)} + \cdots + b_{\pi(n)}) + \cdots +
2d_{\pi(n)}b_{\pi(n)}\\
                    & \leq a_{\pi(1)} + a_{\pi(2)} + \cdots + a_{\pi(n)}\\
\end{array}\right.\\
\end{array}
\end{equation}

The optimal solution to \eqref{eqt1} is described in \eqref{eqt2} and \eqref{eqt3}. Here, $d_{\pi(i)}$ equals to $d(V_{\pi(i)}, V_{\pi(i+1)})$ for $1 \leqslant i \leqslant n-1$. We define a virtual vehicle $V_{\pi(n)}^t$ as a destination point, and $d_{\pi(n)} = d(V_{\pi(n)}, V_{\pi(n)}^t)$.

\begin{equation}\label{eqt2}
D_{\pi} = \sum\limits_{i=1}^n{d_{\pi(i)}} = \sum\limits_{i=1}^n\frac{1}{2}{\Big(a_{\pi(i)} \bigg/ \sum\limits_{j=i}^n{b_{\pi(j)}}\Big)}\end{equation}

\begin{equation}\label{eqt3}
\left\{\begin{array}{llll} & d_{\pi(1)} &= d(V_{\pi(1)}, V_{\pi(2)}) &= \frac{a_{\pi(1)}}{b_{\pi(1)} + \cdots + b_{\pi(n)}} \\
                          & d_{\pi(2)} &= d(V_{\pi(2)}, V_{\pi(3)}) &= \frac{a_{\pi(2)}}{b_{\pi(2)} + \cdots + b_{\pi(n)}} \\
                          & \cdots &\cdots &\cdots \\
                          & d_{\pi(n)} &= d(V_{\pi(n)}, V_{\pi(n)}^t) &= \frac{a_{\pi(n)}}{b_{\pi(n)}} \\
                          \end{array}\right.\\
                          \end{equation}

Since each $D_{\pi}$ is related to a permutation of $n$ vehicles, there supposed to be $n!$ alternatives of $D_{\pi}$, which requires an exponential scale of alternatives before we could be sure of finding the optimal solution.

NVEP can be regarded as a variant of the jeep problem \cite{fine47} involving a convoy of jeeps which travel together, some being used to refuel others, with only one jeep required to travel the maximal distance \cite{phipps47}. A related variant of jeep problem is to determine the range of a fleet of $n$ aircrafts with different fuel capacities and fuel efficiencies. It is assumed that the aircrafts may share fuel in flight and that any of the aircrafts may be abandoned at any stage \cite{franklin60}. A same version of NVEP is to maximize the operational range of a nonidentical vehicle fleet \cite{Mehrez95}. Similarly, considering a fleet of airplanes with refueling technique, the airplane refueling problem \cite{woeginger10} aims to find a optimal drop out sequence to make one of the airplanes travel to the farthest point, which doesn't need all the airplanes fly back to the start point. Literature of such kind of exploration problems suggests that it can be applied in Arctic expeditions, interplanetary travel, such as long distance transportation, multistage rocket fuel, march route, etc.

Previous research posed a special case of NVEP that can be solved by an efficient algorithm, but the computational complexity of general case is still $O(n2^n)$ \cite{yu18}. A multi-task NVEP extended from the single-task NVEP was proved to be NP-hard \cite{xuyy12}. For airplane refueling problem, a polynomial-time approximation scheme was proposed by devising reductions from airplane refueling problem to generalized assignment problem \cite{gamzu19}. However, to our knowledge, there is no literature until now has proved the NP-completeness of NVEP.

In this paper we concentrate our efforts on proving the NP-completeness of NVEP. At first, it is hard for us to reduce a known NP-complete problem \cite{garey79} to NVEP. We finally choose Hamiltonian path to proceeding the reduction because NVEP has a natural graph-based analogue: it could be regarded as a directed graph, then each Hamiltonian path of a directed graph is possible to be equivalent to a sequence of NVEP. Here a Hamiltonian path in a directed graph is a directed path which passes through each vertex exactly once, but need not to return to its starting point \cite{tardos06}. So we prove the NP-completeness of NVEP, and the related problems of NVEP are probably computationally intractable.

\section{Preliminaries}
\label{sec:1}

Given a new problem $X$, the basic strategy for proving its NP-completeness is to choose a known NP-complete problem $Y$ and to prove that $Y \leq_{P} X$. Such a basic reduction strategy can be refined to the following outline of an NP-completeness proof. The refined outline tries to prove that $X$ is at least as hard as $Y$ \cite{tardos06}.

\begin{itemize}
\item[(1)] Prove that $X \in NP$.

\item[(2)] Choose a problem $Y$ that is known to be NP-complete.

\item[(3)] Consider an arbitrary instance $s_Y$ of problem $Y$, and show how to construct, in polynomial time, an instance $s_X$ of problem $X$ that satisfies the following properties:
        \begin{itemize}
        \item[(a)] If $s_Y$ is a "yes" instance of $Y$, then $s_X$ is a "yes" instance of $X$.

        \item[(b)] If $s_X$ is a "yes" instance of $X$, then $s_Y$ is a "yes" instance of $Y$.
        \item[ ] In other words, this establishes that $s_X$ and $s_Y$ have the same answer.
        \end{itemize}

\end{itemize}

The above reduction $Y \leq_{P} X$ consists of transforming a given instance of $Y$ into a single instance of $X$ with the same answer. It requires only a single innovation of the black box to solve $X$. Here, NVEP is problem $X$, and Hamiltonian path is problem $Y$.

\textbf{Problem $X$: decision version of NVEP}

\emph{Instance: Given a fleet of $n$ vehicles $V_1, \cdots, V_n$ with fuel capacities $a_i$ and fuel consumption rates $b_i$ ($1\leq i \leq n$).}

\emph{Question: Is there a sequence $V_{\pi} = (V_{\pi(1)}, \cdots, V_{\pi(n)})$ such that the corresponding travel distance $D_{\pi} \geq n$.}

\textbf{Problem $Y$: decision version of Hamiltonian path}

\emph{Instance: Given a directed graph $G = (V, E)$.}

\emph{Question: Is there a Hamiltonian path $P$ in $G$ that contains each vertex exactly once? (The path is allowed to start at any node and end at any node.)}

\section{Main result}
\label{sec:2}

Given a directed graph $G = (V, E)$ with $n$ nodes, to formulate it to a single NVEP instance, we will regard each node $v_i$ in $G$ as a vehicle $V_i$, and each edge $(v_i, v_j)$ as a consecutive vehicles $(V_i, V_j)$ in a sequence. Here $(V_i, V_j)$ denotes that $V_i \Rightarrow V_j$ in a feasible sequence of NVEP.

\begin{theorem}\label{thm:1} NVEP is NP-complete.\end{theorem}
\begin{proof} It is easy to see that NVEP is in $\mathcal{NP}$: The certificate could be a permutation of the vehicles, and a certifier checks that the sequence contains each vehicle exactly once and the length of the corresponding trip is at least $n$.

We now show that Hamiltonian path $\leq_P$ NVEP, and we aim to show that any Hamiltonian path can be reduced in polynomial time to a feasible sequence of NVEP with length at least $n$.

Given a directed graph $G = (V, E)$ with $n$ nodes, we define $n$ vehicles with $a_i = i$ and $b_i = 1/2$ for $1 \leqslant i \leqslant n$. Such a fleet of $n$ vehicles are used to construct the reduced NVEP instance. For each node $v_i$ in $G$, there is a vehicle $V_i$ in NVEP. If there is an edge $e = (v_i, v_j)$ in $G$, we define $d_i = d(V_i, V_j) > 0$ and $V_i \Rightarrow V_j$ in a feasible sequence $V_{\pi}$ of NVEP. Otherwise, we define $d_i = d(V_i, V_j) = 0$ and $(V_i, V_j)$ can not be consecutive vehicles in a feasible sequence. Besides, we add some virtual nodes $v_i^{t}$ to $G$, and we define some virtual vehicles $V_i^{t}$ to NVEP. For each virtual node $v_i^{t}$ as an end node of Hamiltonian path, there is an edge between $v_i$ and $v_i^{t}$. Similarly for each vehicle $V_i^{t}$ as a destination point, we define $d_i = d(V_i, V_i^t) > 0$ and $V_i \Rightarrow V_i^{t}$ for $1 \leqslant i \leqslant n$.

Now we claim that the directed graph $G$ has a Hamiltonian path if and only if there is a feasible sequence of length at least $n$ in the reduced NVEP instance.

If $G$ has a Hamiltonian path $P_{\pi} = (v_{\pi(1)}, v_{\pi(2)}, \cdots, v_{\pi(n)}, v_{\pi(n)}^{t})$, then this order of the corresponding vehicles define a feasible sequence $V_{\pi} = (V_{\pi(1)}, V_{\pi(2)}, \cdots, V_{\pi(n)}, V_{\pi(n)}^{t})$ of length at least $n$. So a given Hamiltonian path $P_{\pi}$ is "spread" to the reduced NVEP instance by establishing a sequence $V_{\pi}$, which is also a feasible sequence of NVEP with total distance $D_{\pi}$ at least $n$.

Conversely, suppose there is a feasible sequence $V_{\pi}$ of the reduced NVEP instance with total distance at least $n$. The expression for the total distance $D_{\pi}$ of this feasible sequence is a sum of $n$ terms, each of which is greater than $0$; thus the feasible sequence must contain all the vehicles and has $n$ pairs of consecutive vehicles with $d(V_{\pi(i)}, V_{\pi(i+1)}) > 0$ for $1 \leqslant i \leqslant n-1$, and $d(V_{\pi(n)}, V_{\pi(n)}^t) > 0$. Hence each pair of nodes in $G$ that correspond to the $n$ pairs of consecutive vehicles in the feasible sequence must be connected by an edge. It follows that the ordering of these corresponding nodes must form a Hamiltonian path. \end{proof}

%If $D_{\pi} = n$, then the sequence consists of $n$ segmented parts, each of which has to be equal to $1$ to make sure that the total distance is equal to $n$; If $D_{\pi} > n$, we can set $a_{\pi(i)} = n + 1 -i$ and $b_{\pi(i)} = 1/2$ for each $1 \leq i \leq n$. Then the $D_{\pi}$ is still settled to be equal to $n$.
\section{Conclusion}
\label{sec:3}

Proving the NP-completeness of NVEP and its equivalent nonlinear optimization problems has shown particular difficulties before. The main difficulty is obvious, which is caused by the fractional form of objective function and what reduction to use. We prove the NP-completeness of NVEP by using the reduction "Hamiltonian path $\leq_{P}$ NVEP", which consists of encoding a given instance of Hamiltonian path into an instance of NVEP with the same answer. The reduction is refined from the basic reduction strategy to resort to Hamiltonian path as a known NP-complete problem. Consequently we prove that NVEP is NP-complete.

%\acknowledgements{\rm Thanks $\cdots$}
%% Please thank the anonymous people who make contributions to this article. If you don't want it, please delete it.

\end{document}